\documentclass[11pt,letterpaper]{article}  
\usepackage{amsmath,amssymb,amsfonts,amsthm}
\usepackage{cite}
\usepackage{url}
\newtheorem{theorem}{Theorem}
\newtheorem{corollary}[theorem]{Corollary}
\newtheorem{lemma}[theorem]{Lemma}

\theoremstyle{definition}

\newcommand{\bigO}{\mathcal{O}}
\newcommand{\expect}{\mathbb{E}}
\newcommand{\prob}{\mathbb{P}}
\newcommand{\reals}{\mathbb{R}}

\renewcommand{\epsilon}{\varepsilon}
\renewcommand{\leq}{\leqslant}
\renewcommand{\geq}{\geqslant}
\newcommand{\fGr}[1][r]{{f_{G,{#1}}}}
\newcommand{\fbar}{{\bar{f}}}
 
 \title {Approximating Fixation Probabilities in\\
        the Generalized Moran Process%
        \thanks{A preliminary version of this work appeared in
                \emph{Proceedings of the ACM--SIAM Symposium on Discrete
                Algorithms (SODA)}, pages~954--960, 2012.}}

\author{Josep D{\'i}az%
        \thanks{\protect\raggedright
                Departament de Llenguatges i Sistemes Inform{\'a}tics,
                Universitat Polit{\'e}cnica de Catalunya, Spain. 
                Email: \{\texttt{diaz},
                            \texttt{mjserna}\}\texttt{@lsi.upc.edu}\,.},
        Leslie Ann Goldberg%
        \thanks{\protect\raggedright
                Department of Computer Science, University of Liverpool,
                UK.  Email: \{\texttt{L.A.Goldberg},
                \texttt{David.Richerby}\}\texttt{@liverpool.ac.uk}.
                Supported by EPSRC grant EP/I011528/1
                    \emph{Computational Counting}.},
        George B.~Mertzios%
        \thanks{\protect\raggedright
                School of Engineering and Computing Sciences,
                Durham University, UK.
                Email: \texttt{george.mertzios@durham.ac.uk}\,.},\\
        David Richerby\footnotemark[3], 
        Maria Serna\footnotemark[2]\ \ and 
        Paul G.~Spirakis%
        \thanks{\protect\raggedright
                Department of Computer Engineering and Informatics,
                University of Patras, Greece.
                Email: \texttt{spirakis@cti.gr}\,.}}

\date{\vspace{-0.8cm}}

\begin{document}
\maketitle{}

\begin{abstract}
    \noindent We consider the Moran process, as generalized by
    Lieberman, Hauert and Nowak (\emph{Nature}, 433:312--316, 2005).
    A population resides on the vertices of a finite, connected, undirected
    graph and, at each time step, an individual is chosen at random
    with probability proportional to its assigned ``fitness'' value.
    It reproduces, placing a copy of itself on a neighbouring vertex
    chosen uniformly at random, replacing the individual that was
    there.  The initial population consists of a single mutant of
    fitness $r>0$ placed uniformly at random, with every other vertex
    occupied by an individual of fitness~1.  The main quantities of
    interest are the probabilities that the descendants of the initial
    mutant come to occupy the whole graph (fixation) and that they die
    out (extinction); almost surely, these are the only possibilities.
    In general, exact computation of these quantities by standard
    Markov chain techniques requires solving a system of linear
    equations of size exponential in the order of the graph so is not
    feasible.  We show that, with high probability,
    the number of steps needed to reach fixation or extinction is
    bounded by a polynomial in the number of vertices in the graph.
    This bound allows us
    to construct fully polynomial randomized approximation schemes
    (FPRAS) for the probability of fixation (when $r\geq 1$) and of
    extinction (for all $r>0$).

\end{abstract}
 \section{Introduction}
\label{sec:Intro}

Population and evolutionary dynamics have been extensively studied
\cite{Traulsen09, Taylor06, Taylor04, Ohtsuki06, Karlin75, Broom09,
  Antal06}, usually with the assumption that the evolving population
has no spatial structure.  One of the main models in this area is the Moran
process~\cite{Moran58}.  The initial population contains
a single ``mutant'' with fitness $r>0$, with all other individuals
having fitness~1.  At each step of the process, an individual is
chosen at random, with probability proportional to its fitness.  This
individual reproduces, replacing a second individual, chosen uniformly
at random, with a copy of itself.  Population dynamics has also
been studied in the context of strategic interaction in evolutionary
game theory~\cite{Gintis00, Sandholm11, Hofbauer98, Kandori93,
  Imhof05}.

Lieberman, Hauert and Nowak~\cite{Nowak05, nowak06-book} introduced a
generalization of
the Moran process, where the members of the population are placed on the
vertices of a connected graph which is, in general,
directed.  In this model, the initial
population again consists of a single mutant of fitness $r>0$ placed
on a vertex chosen uniformly at random, with each other vertex
occupied by a non-mutant with fitness~1.  The individual that will
reproduce is chosen as before but now one of its neighbours is
randomly selected for replacement, either uniformly or according to a
weighting of the edges.  The original Moran process can be recovered
by taking the graph to be an unweighted complete graph.  In the present paper,
we consider the process on finite, unweighted, undirected graphs.

Several similar models describing particle interactions have been
studied previously, including the SIR and SIS epidemic
models~\cite[Chapter~21]{EasleyK10}, the voter model, the antivoter
model and the exclusion process~\cite{Liggett85, Aldous-online-book,
  Durrett88}.  Related models, such as the decreasing cascade
model~\cite{kempe05, Mossel07}, have been studied in the context of
influence propagation in social networks and other models have been
considered for dynamic monopolies~\cite{Berger01}.  However, these
models do not consider different fitnesses for the individuals.

In general, the Moran process on a connected, directed graph may end
with all vertices occupied by mutants or with no vertex occupied by a
mutant --- these cases are referred to as \emph{fixation} and
\emph{extinction},
respectively --- or the process may continue forever.  However, for
finite undirected graphs and finite strongly connected digraphs, the
process terminates almost surely, either at fixation or extinction.
At the other extreme, fixation is
impossible in the directed graph with vertices $\{x,y,z\}$ and edges
$\{\overrightarrow{xz}, \overrightarrow{yz}\}$ and extinction is
impossible unless the mutant starts at $z$.  The \emph{fixation
  probability} for a mutant of fitness $r$ in a graph $G$ is the
probability that fixation is reached and is denoted $\fGr$.

The fixation probability can, in principle, be determined by standard
Markov chain
techniques.  However, doing so for a general graph on $n$ vertices
requires solving a set of $2^n$ linear equations, which is not
computationally feasible, even numerically.  As a result, most prior
work on computing
fixation probabilities in the generalized Moran process has either
been restricted to small
graphs~\cite{Broom09} or graph classes where a high degree of symmetry
reduces the size of the set of equations --- for example, paths,
cycles, stars and complete graphs~\cite{Broom08, Broom-two-results, Broom10}
--- or has concentrated on finding graph classes that either encourage
or suppress the spread of the mutants~\cite{Nowak05, MNRS:Evolution}.
Rycht\'a\v{r} and Stadler present some experimental results on fixation
probabilities for random graphs derived from
grids~\cite{RS2008:Smallworlds}.

Because of the apparent intractability of exact computation, we turn to
approximation.  Using a potential function argument, we show that,
with high probability, the Moran process on an undirected graph of
order $n$, with a single initial mutant chosen uniformly at random,
 reaches absorption (either fixation or extinction) within
$\bigO(n^6)$ steps if $r=1$ and $\bigO(n^4)$ and $\bigO(n^3)$ steps
when $r>1$ and $r<1$, respectively.  Taylor et al.~\cite{Taylor06}
studied absorption times for variants of the generalized Moran process
but, in our setting, their results only apply to the process on
regular graphs, where it is equivalent to a biased random walk on a
line with absorbing barriers.  The absorption time analysis of Broom
et al.~\cite{Broom10} is also restricted to complete graphs, cycles and stars.
In contrast to this earlier work, our results apply to all connected
undirected graphs.

When $r=1$, we show that the fixation probability is $\tfrac{1}{n}$ on
any connected $n$-vertex graph.  For $r\neq 1$, our bound on the
absorption time, along with polynomial upper and
lower bounds for the fixation probability, allows the estimation of
the fixation and extinction probabilities by Monte Carlo techniques.
Specifically, we give a \emph{fully polynomial randomized
  approximation scheme} (FPRAS) for these quantities.  An FPRAS for a
function $f(X)$ is a polynomial-time randomized algorithm $g$ that,
given input $X$ and an error bound $\epsilon$ satisfies
$(1-\epsilon)f(X) \leq g(X) \leq (1+\epsilon)f(X)$ with probability at
least $\frac{3}{4}$ and runs in time polynomial in the length of $X$
and $\frac{1}{\epsilon}$~\cite{KL1983:Monte-Carlo}.  (The probability
can be ``boosted'' to any value in $[\tfrac{3}{4},1)$ at small cost
  \cite{JVV1986:Randgen}.)

For the case $r<1$, the fixation probability may be exponentially
small (see Section~\ref{sec:Pfix}).  As a result, there is no positive
polynomial lower bound on the fixation probability so only the
extinction probability can be approximated by the above Monte Carlo
technique.  (Note that, when $f \ll 1$, computing
$1-f$ to within a factor of $1\pm \epsilon$ does not imply computing
$f$ to within the same factor.)

\paragraph{Notation.} We consider only finite, connected, undirected
graphs $G=(V,E)$ and we write $n = |V|$ (the \emph{order} of the
graph).  Our results apply only to connected graphs as, otherwise, the
fixation probability is necessarily zero; we also exclude the
one-vertex graph to avoid trivialities.  The edge between vertices $x$
and $y$ is denoted by $xy$.  For a subset $X\subseteq V(G)$, we write
$X+y$ and $X-y$ for $X\cup\{y\}$ and $X\setminus\{y\}$, respectively.

Throughout, $r$ denotes the fitness of the mutants.
A \emph{state} of the Moran process is the set
of vertices occupied by mutants at a given time.  The \emph{total
fitness} of the state $S\subseteq V(G)$ is $W(S)=r|S| + |V\setminus S|$.
We write $\fGr(S)$ for the fixation probability of
$G$, when the initial state is $S$ and, for $x\in V(G)$, we
write $\fGr(x)$ for $\fGr(\{x\})$.  We denote by $\fGr = \frac{1}{n}
\sum_{x\in V} \fGr(x)$ the \emph{fixation probability} of $G$; that
is, the probability that a single mutant with fitness $r$ placed
uniformly at random in $V$ eventually takes over the graph $G$.
The \emph{absorption time} of a Moran process $(X_i)_{i\geq 0}$ is the
random variable $\tau = \min\,\{i\mid X_i=0 \text{ or } X_i=V(G)\}$.
Finally, we define the problem \textsc{Moran fixation} (respectively,
\textsc{Moran extinction}) as follows: given a graph $G=(V,E)$ and a
fitness value $r>0$, compute the value $\fGr$ (respectively,
$1-\fGr$).

\paragraph{Organization of the paper.}  In Section~\ref{sec:Pfix}, we
investigate the fixation probability when $r=1$ and demonstrate
polynomial upper and lower bounds for $\fGr$ for any $r\geq 0$.  In
Section~\ref{sec:Absorb}, we use our potential function to derive
polynomial bounds on the absorption time (both in expectation and with
high probability) in general undirected graphs.  Our FPRAS for
computing fixation and extinction probabilities appears in
Section~\ref{sec:FPRAS}.

 \section{Bounding the fixation probability}
\label{sec:Pfix}

Lieberman et al.\ \cite{Nowak05} (see also~\cite[p.~135]{nowak06-book})
observed that, if $G$ is a directed graph with a
single source (a vertex with in-degree zero), then $\fGr =
\frac{1}{n}$, independent of the fitness of the mutants.  We first
show that, when $r=1$, the fixation probability is also
$\tfrac{1}{n}$, independent of the graph structure.

\begin{lemma}
\label{lemma:fG-one}
    Let $G=(V,E)$ be an undirected graph with $n$ vertices. Then
    $\fGr[1] = \frac{1}{n}$.
\end{lemma}
\begin{proof}
    Consider the variant of the process where every vertex starts with
    its own colour and every vertex has fitness~1.  Allow the process
    to evolve as usual: at each step, a vertex is chosen uniformly at
    random and its colour is propagated to a neighbour also chosen
    uniformly at random.  At any time, we can consider the vertices of
    any one colour to be the mutants and all the other vertices to be
    non-mutants.  Hence, with probability~1, some colour will take
    over the graph and the probability that $x$'s initial colour takes
    over is exactly $\fGr[1](x)$.  Thus, $\fGr[1] = \frac{1}{n}
    \sum_{x\in V} \fGr[1](x) = \frac{1}{n}$.
\end{proof}

This allows us to give a lower bound on the fixation probability
whenever $r\geq 1$.

\begin{corollary}
\label{cor:fG-lowerbound}
    Let $G=(V,E)$ be an undirected graph with $n$ vertices. Then $\fGr
    \geq \frac{1}{n}$ for any $r\geq 1$.
\end{corollary}
\begin{proof}
    By \cite[Theorem~6]{SR2011:FastFixProb}, $\fGr\geq \fGr[1]$ for
    any $r\geq 1$.
\end{proof}

The process used to prove Lemma~\ref{lemma:fG-one} also gives a short
and direct proof of the following result of Shakarian and
Roos~\cite[Theorem~5]{SR2011:FastFixProb}.

\begin{corollary}
    For $S_1,S_2\subseteq V(G)$,
    \begin{equation*}
        \fGr[1](S_1\cup S_2) = \fGr[1](S_1) + \fGr[1](S_2)
                               - \fGr[1](S_1\cap S_2)\,.
    \end{equation*}
\end{corollary}
\begin{proof}
    As before, we give each vertex its own colour and fitness~1 and,
    at any point, we can consider any subset of the colours to be the
    mutants.  For any $S\subseteq V$, $\fGr[1](S)$ is the probability
    that, eventually, only colours from $S$ remain.  With
    probability~1, a single colour will then take over, so $\fGr[1](S)
    = \sum_{x\in S} \fGr[1](x)$ and the result is immediate.
\end{proof}

Note that there is no polynomial lower bound corresponding to
Corollary~\ref{cor:fG-lowerbound} when $r<1$.  For example, for $r\neq
1$, the fixation probability of the complete graph $K_n$ is given by
\begin{equation*}
    f_{K_n,r} = \frac{1-\frac{1}{r}}{1-\frac{1}{r^n}}\,.
\end{equation*}
For $r>1$, this is at least $1-\frac{1}{r}$ but there is
no positive polynomial lower bound where $r<1$.

\begin{lemma}
\label{lemma:fG-upperbound}
    Let $G=(V,E)$ be an undirected graph with $n$ vertices. Then
    $\fGr\leq 1 - \frac{1}{n+r}$ for any $r > 0$.
\end{lemma}
\begin{proof}
    For any vertex $x\in V\!$, let $Q(x) = \sum_{xy\in E}
    \frac{1}{\deg y}$, so $\sum_{x\in V} Q(x) = n$.

    To give an upper bound for $\fGr(x)$ for every $x\in V\!$, we
    relax the Markov chain by assuming that fixation is reached as
    soon as a second mutant is created.  From the state $S = \{x\}$,
    the probability that a new mutant is created is $a(x) =
    \frac{r}{n-1+r}$ and the probability that one of $x$'s non-mutant
    neighbours reproduces into $x$ is $b(x) = \frac{1}{n-1+r}Q(x)$.
    The probability that the population stays the same, because a
    non-mutant reproduces to a non-mutant vertex, is $1-a(x)-b(x)$.
    The probability that the mutant population reaches two (i.e., that
    the first change to the state is the creation of a new mutant) is
    given by
    \begin{equation*}
        p(x) = \frac{a(x)}{a(x) + b(x)} = \frac{r}{r + Q(x)}\,.
    \end{equation*}
    Therefore, the probability that the new process reaches fixation is
    \begin{equation*}
        p = \frac{1}{n} \sum_{x\in V} p(x)
                         = \frac{r}{n} \sum_{x\in V}\frac{1}{r + Q(x)}\,.
    \end{equation*}

    Writing $p = \frac{r}{n}\sum_{i=1}^n(r + q_i)^{-1}\!$, we wish to
    find the maximum value of $p$ subject to the constraints that $q_i>0$ for
    all $i$ and $\sum_{i=1}^n q_i = \sum_{x\in V} Q(x) = n$.  If we relax
    the first constraint to $q_i\geq 0$, the sum is maximized by
    setting $q_1 = n$ and $q_2 = \dots = q_n = 0$.  Therefore,
    \begin{equation*}
        \fGr \leq p\leq \frac{r}{n} \left( \frac{1}{r+n}
                                            + (n-1)\frac{1}{r+0} \right)
            = 1 - \frac{1}{r+n}\,.\qedhere
    \end{equation*}
\end{proof}

 \section{Bounding the absorption time}
\label{sec:Absorb}

In this section, we show that the Moran process on a connected graph
$G$ of order $n$ is expected to reach absorption in a polynomial
number of steps.  To do this, we use the potential function given by
\begin{equation*}
    \phi(S) = \sum_{x\in S} \frac{1}{\deg x}
\end{equation*}
for any state $S\subseteq V(G)$ and we write $\phi(G)$ for
$\phi(V(G))$.  Note that $1 < \phi(G) < n$ and that $\phi(\{x\}) =
1/\deg x\leq 1$ for any vertex $x\in V$.

First, we show that the potential strictly increases in expectation
when $r > 1$ and strictly decreases in expectation when $r<1$.

\begin{lemma}
\label{lemma:phi-bound}
    Let $(X_i)_{i\geq 0}$ be a Moran process on a graph $G=(V,E)$
    and let $\emptyset\subset S \subset V\!$.  If $r \geq 1$, then
    \begin{equation*}
        \expect[\phi(X_{i+1}) - \phi(X_i) \mid X_i=S] \geq
        \left(1 - \frac{1}{r}\right) \cdot \frac{1}{n^3}\,,
    \end{equation*}
    with equality if, and only if, $r=1$. For $r<1$,
    \begin{equation*}
        \expect[\phi(X_{i+1}) - \phi(X_i) \mid X_i=S] < \frac{r-1}{n^3}\,.
    \end{equation*}
\end{lemma}
\begin{proof}
    Write $W(S) = n + (r-1)|S|$ for the total fitness of the
    population.  For $\emptyset \subset S \subset V\!$, and any
    value of $r$, we have
    {
    \newcommand{\thesumof}[1]{\frac{{#1}}{W(S)}\!\!\sum_{\substack{xy\in E \\ x\in S, y\in \overline{S}}}\!\!{}}
    \begin{align}
        \hspace{4em}&\hspace{-4em}
        \expect[\phi(X_{i+1}) -\phi(X_i) \mid  X_i = S] \notag\\
            &= \thesumof{1} \left(
                   r\cdot \frac{\phi(S + y) - \phi(S)}{\deg x}
                   + \frac{\phi(S-x) - \phi(S)}{\deg y}\right)\notag\\
            &= \thesumof{1} \left(
                   r \cdot \frac{1}{\deg y}\cdot\frac{1}{\deg x}
                   - \frac{1}{\deg x}\cdot\frac{1}{\deg y}\right)\notag\\
    \label{eq:deltaphi}
            &= \thesumof{r-1}
                    \frac{1}{\deg x \deg y}\,.
    \end{align}}

    This is clearly zero if $r=1$.  Otherwise, the sum is minimized in
    absolute value by noting that there must be at least one
    edge between $S$ and $\overline{S}$ and that its endpoints have
    degree at most $(n-1) < n$.  The greatest-weight state is
    the one with all mutants if $r > 1$ and the one with no mutants
    if $r < 1$.  Therefore, if $r > 1$, we have
    \begin{equation*}
        \expect[\phi(X_{i+1}) -\phi(X_i) \mid  X_i = S]
            \ >\ \frac{r-1}{rn}\cdot \frac{1}{n^2}
            \ =\ \left(1 - \frac{1}{r}\right)\frac{1}{n^3}
    \end{equation*}
    and, if $r<1$,
    \begin{equation*}
        \expect[\phi(X_{i+1}) -\phi(X_i) \mid  X_i = S]
            \ <\ (r-1) \frac{1}{n^3}\,.\qedhere
    \end{equation*}
\end{proof}

The method of bounding in the above proof appears somewhat crude ---
for example, in a graph of order $n>2$, if both endpoints of the chosen
edge from $S$ to $\overline{S}$ have degree $n-1$ then there must be
more edges between mutants and non-mutants.  Nonetheless, over the
class of all graphs, the bound of Lemma~\ref{lemma:phi-bound} is
asymptotically optimal up to constant factors.  For $n\geq 2$, let
$G_n$ be the $n$-vertex graph made by adding an edge between the centres of two
disjoint stars of as close-to-equal size as possible.  If $S$ is the
vertex set of one of the stars, $\expect[\phi(X_{i+1}) - \phi(X_i)
  \mid X_i=S] = \Theta(n^{-3})$.

However, it is possible to specialize equation~\eqref{eq:deltaphi}
to give better
bounds for restricted classes of graphs.  For example, if we consider
graphs of bounded degree then $(\deg x \deg y)^{-1} = \bigO(1)$ and
the expected change in $\phi$ is $\bigO(\frac{1}{n})$.

To bound the expected absorption time, we use martingale techniques.
It is well known how to bound the expected absorption time using a
potential function that decreases in expectation until absorption.
This has been made explicit by Hajek~\cite{Haj1982:Hitting-time-drift}
and we use the following formulation based on that of He and
Yao~\cite{HY2001:Drift}.  The proof is essentially theirs but modified
to give a slightly stronger result.

\begin{theorem}
\label{thrm:He-Yao}
    Let $(Y_{\!i})_{i\geq 0}$ be a Markov chain with state space
    $\Omega$, where $Y_{\!0}$ is chosen from some set $I\subseteq
    \Omega$.  If there are constants $k_1,k_2>0$ and a non-negative
    function $\psi\colon \Omega\to \reals$ such that
    \begin{itemize}
    \item $\psi(S) = 0$ for some $S\in\Omega$,
    \item $\psi(S)\leq k_1$ for all $S\in I$ and
    \item $\expect[\psi(Y_i) - \psi(Y_{i+1}) \mid Y_i = S] \geq k_2$ for
          all $i\geq 0$ and all $S$ with $\psi(S)>0$,
    \end{itemize}
    then $\expect[\tau] \leq k_1/k_2$, where $\tau = \min\,\{i:
    \psi(Y_i) = 0\}$.
\end{theorem}
\begin{proof}
    By the third condition, the chain is a supermartingale so
    it converges to zero almost surely
    \cite[Theorem~II-2-9]{Nev1975:Martingales}.
    \begin{align*}
        \expect[\psi(Y_{\!i}) &\mid \psi(Y_{\!0}) > 0 ] \\
            &\ =\ \expect\big[
                  \expect\big[\psi(Y_{\!i-1}) +
                      \big(\psi(Y_{\!i}) - \psi(Y_{\!i-1})\big)
                            \mid Y_{\!i-1}\big]
                   \mid  \psi(Y_{\!0}) > 0\big] \\
            &\ \leq\ \expect[\psi(Y_{\!i-1}) - k_2 \mid \psi(Y_{\!0}) > 0]\,.
    \end{align*}
    Induction on $i$ gives $\expect[\psi(Y_{\!i}) \mid \psi(Y_{\!0}) >
    0] \leq \expect[\psi(Y_0) - ik_2 \mid \psi(Y_{\!0}) > 0]$ and,
    from the definition of the stopping time $\tau$,
    \begin{align*}
        0\ &=\ \expect[\psi(Y_{\!\tau}) \mid \psi(Y_{\!0}) > 0] \\
           &\leq\ \expect[\psi(Y_{\!0})]
              - k_2\expect[\tau \mid \psi(Y_{\!0}) > 0 ]\\
           &\leq\ k_1 - k_2 \expect[\tau \mid \psi(Y_{\!0}) > 0]\,.
    \end{align*}

    The possibility that $\psi(Y_{\!0}) = 0$ can only decrease the
    expected value of $\tau$ since, in that case, $\tau = 0$.
    Therefore, $\expect[\tau] \leq \expect[\tau \mid \psi(Y_{\!0}) >
    0] \leq k_1/k_2$.
\end{proof}

\begin{theorem}
\label{thrm:abs-time-disadv}
    Let $G = (V,E)$ be a graph of order $n$.  For $r<1$ and any
    $S\subseteq V$, the absorption time $\tau$ of the Moran process on
    $G$ satisfies
    \begin{equation*}
        \expect[\tau \mid X_0=S] \leq \frac{1}{1-r}n^3 \phi(S)\,.
    \end{equation*}
\end{theorem}
\begin{proof}
    Let $(Y_i)_{i\geq 0}$ be the process on $G$ that behaves
    identically to the Moran process except that, if the mutants reach
    fixation, we introduce a new non-mutant on a vertex chosen
    uniformly at random.  That is, from the state $V\!$, we move to
    $V-x$, where $x$ is chosen uniformly at random, instead of staying in $V\!$.
    Writing $\tau' = \min\,\{i: Y_i = \emptyset\}$ for the absorption time
    of this new process, it is clear that $\expect[\tau \mid X_0=S] \leq
    \expect[\tau' \mid Y_0=S]$.

    The function $\phi$ meets the criteria for $\psi$ in the statement
    of Theorem~\ref{thrm:He-Yao} with $k_1 = \phi(S)$ and $k_2 =
    (1-r)n^{-3}\!$.  The first two conditions of the theorem are
    obviously satisfied.  For $Y_i\subset V\!$, the third condition is
    satisfied by Lemma~\ref{lemma:phi-bound} and we have
    \begin{equation*}
        \expect[\phi(Y_i) - \phi(Y_{i+1}) \mid Y_i = V]
            = \frac{1}{n}\sum_{x\in V} \frac{1}{\deg x}
            > \frac{1}{n} > k_2\,.
    \end{equation*}
    Therefore, $\expect[\tau \mid X_0=S] \leq \expect[\tau' \mid
    Y_0=S] \leq \frac{1}{1-r}n^3\phi(S)$.
\end{proof}

\begin{corollary}
\label{cor:abs-time-disadv-whp}
    Let $G = (V,E)$ be a graph of order $n$.  For $r<1$ and when the
    initial single mutant is chosen uniformly at random, the
    absorption time $\tau$ of the Moran process on $G$ satisfies
    \begin{equation*}
        \expect[\tau] \leq \frac{1}{1-r}n^3.
    \end{equation*}
    Further, the process reaches absorption within $t$ steps with
    probability at least $1-\epsilon$, for any $\epsilon\in(0,1)$ and
    any $t\geq \frac{1}{1-r} n^3 / \epsilon$.
\end{corollary}
\begin{proof}
    For the first part,
    \begin{equation*}
        \expect[\tau]
            \leq \sum_{x\in V} \frac{1}{n}
                     \cdot \frac{1}{1-r}n^3 \phi(\{x\})
            \leq \frac{1}{1-r}n^3
    \end{equation*}
    and the second part is immediate from Markov's inequality.
\end{proof}

For $r>1$, the proof needs slight adjustment because, in this case,
$\phi$ increases in expectation.

\begin{theorem}
\label{thrm:abs-time}
    Let $G = (V,E)$ be a graph of order $n$.  For $r>1$ and any
    $S\subseteq V$, the absorption time $\tau$ of the Moran process on
    $G$ satisfies
    \begin{equation*}
        \expect[\tau \mid X_0=S]
            \leq \frac{r}{r-1}n^3 \big(\phi(G) - \phi(S)\big)
            \leq \frac{r}{r-1}n^4.
    \end{equation*}
\end{theorem}
\begin{proof}
    Let $(Y_i)_{i\geq 0}$ be the process that behaves identically to
    the Moran process $(X_i)_{i\geq 0}$ except that, if $Y_j =
    \emptyset$, then $Y_{j+1} = \{x\}$, where $x$ is a vertex chosen
    uniformly at random.  Setting $\tau' = \min\,\{i: Y_i = V\}$, we
    have $\expect[\tau | X_0=S] \leq \expect[\tau' | Y_0=S]$.

    Putting $\psi(Y) = \phi(G) - \phi(Y)$, $k_1 = \psi(S)\leq n$ and
    $k_2 = (1-\frac{1}{r})n^{-3}$ satisfies the conditions of
    Theorem~\ref{thrm:He-Yao} --- the third condition follows from
    Lemma~\ref{lemma:phi-bound} for $\emptyset \subset Y_i \subset V$
    and
    \begin{equation*}
        \expect[\psi(Y_i) - \psi(Y_{i+1}) \mid Y_i = \emptyset]
            = \frac{1}{n}\sum_{x\in V} \frac{1}{\deg x}
            > \frac{1}{n} > k_2.
    \end{equation*}
    The result follows from Theorem~\ref{thrm:He-Yao}.
\end{proof}

\begin{corollary}
\label{cor:abs-time-whp}
    When $r>1$ and the initial single mutant is chosen uniformly at
    random, the absorption time of the Moran process on an $n$-vertex
    graph $G$ satisfies
    \begin{equation*}
        \expect[\tau] \leq \frac{r}{r-1}n^4.
    \end{equation*}
    Further, the process reaches absorption within $t$ steps with probability
    at least $1-\epsilon$, for any $\epsilon\in(0,1)$ and any $t\geq
    \frac{r}{r-1} n^3 \phi(G) / \epsilon$.
\end{corollary}
\begin{proof}
    The first part follows from the theorem and the fact that $\phi(G)
    - \phi(\{x\}) \leq n - \tfrac{1}{n} < n$ for any vertex $x$.
    The second part is by Markov's inequality.
\end{proof}

The $\bigO(n^4)$ bound in Corollary~\ref{cor:abs-time-whp} does not seem to
be very tight and could, perhaps, be improved by a more careful
analysis, which we leave for future work.  In simulations, every class
of graphs we have considered has had expected fixation time
$\bigO(n^3)$ for $r>1$.  The graphs $G_n$ described after
Lemma~\ref{lemma:phi-bound} are the slowest we have found but, even on
those graphs, the absorption time is, empirically, still $\bigO(n^3)$.
Note that $n-2 < \phi(G_n) < n-1$ so, for these graphs, even the bound
of $\frac{r}{r-1}n^3\phi(G_n)$ is $\bigO(n^4)$.

The case $r=1$ is more complicated as Lemma~\ref{lemma:phi-bound}
shows that the expectation is constant.  However, this allows us to
use standard martingale techniques and the proof of the following is
partly adapted from the proof of Lemma~3.4 in
\cite{LRS2001:Markov-planar-lattice}.

\begin{theorem}
\label{thrm:abs-martingale}
    The expected absorption time for the Moran process
    $(X_i)_{i\geq 0}$ with $r=1$ on a graph $G=(V,E)$ is at most
    $n^4(\phi(G)^2 - \expect[\phi(X_0)^2])$.
\end{theorem}
\begin{proof}
    Let $m = \phi(G)/2$ and let $\psi_i = m - \phi(X_i)$.  Thus,
    $-m\leq \psi_i\leq m$ for all $i$.
    By Lemma~\ref{lemma:phi-bound}, $\expect[\phi(X_{i+1}) \mid X_i] \geq
    \phi(X_i)$ so
    \begin{equation}
    \label{eq:psi-decreases}
        \expect[\psi_{i+1} \mid X_i] \leq \psi_i\,.
    \end{equation}
    (In fact, by Lemma~\ref{lemma:phi-bound}, $\expect[\psi_{i+1} \mid
      X_i] = \psi_i$ but we do not need this.)

    From the definition of the process, $\psi_{i+1}\neq \psi_i$ if,
    and only if, $X_{i+1}\neq X_i$.  Therefore, $\prob[\psi_{i+1}\neq
      \psi_i] = \prob[X_{i+1}\neq X_i]$ and, for $0 < |X_i| < n$, this
    probability is at least $n^{-2}$ because there is at least one
    edge from a mutant to a non-mutant.  From the definition of
    $\phi$, if $\psi_{i+1}\neq \psi_i$ then $|\psi_{i+1}-\psi_i|\geq
    n^{-1}\!$.  When $|\psi_i| < m$, it follows that
    \begin{equation}
    \label{eq:psi-squared}
        \expect[(\psi_{i+1}-\psi_i)^2\mid X_i] \geq n^{-4}\,.
    \end{equation}

    Let $t_0 = \min\,\{t: |\psi_t| = m\}$, which is a stopping time for
    the sequence $(\psi_t)_{t\geq 0}$ and is also the least $t$ for
    which $X_t = \emptyset$ or $X_t = V\!$.  Let
    \begin{equation*}
        Z_t = \begin{cases}
               \ \psi_t^2 - 2m\psi_t - n^{-4}t
                                             &\text{if } |\psi_t| < m \\
               \ 3m^2 - n^{-4}t_0            &\text{otherwise.}
               \end{cases}
    \end{equation*}

    We now show that $(Z_t)_{t\geq 0}$ is a submartingale.  This is
    trivial for $t\geq t_0$, since then we have $Z_{t+1} = Z_t$.  In
    the case where $t < t_0$,
    \begin{align*}
        \hspace{2em}&\hspace{-2em}
        \expect[Z_{t+1} - Z_t \mid X_t] \\
            &\geq\ \expect[\,\psi_{t+1}^2 - 2m\psi_{t+1} - n^{-4}(t+1)
                           - \psi_t^2 + 2m\psi_t + n^{-4}t \mid X_t] \\
            &=\ \expect[-2m(\psi_{t+1} - \psi_t)
                            + \psi_{t+1}^2 - \psi_t^2 - n^{-4} \mid X_t] \\
            &=\ \expect[\,2(\psi_t - m)(\psi_{t+1} - \psi_t)
                            + (\psi_{t+1} - \psi_t)^2 - n^{-4} \mid X_t] \\
            &\geq\ 0\,.
    \end{align*}
    The first inequality is because $3m^2\geq \psi_t^2 - 2m\psi_t$ for
    all $t$, since $|\psi_t|\leq m$.  The final inequality comes from
    equations \eqref{eq:psi-decreases} and~\eqref{eq:psi-squared}.
    Note also that $\expect[Z_{t+1} - Z_t\mid X_t]\leq
    6m^2 < \infty$ in all cases.

    We have
    \begin{equation*}
        \expect[Z_0]
            = \expect\big[
                  \big(m - \phi(X_0)\big)^2 - 2m\big(m-\phi(X_0)\big)
              \big] \\
            = \expect[\phi(X_0)^2] - m^2
    \end{equation*}
    and $\expect[Z_{t_0}] = 3m^2 - n^{-4}\expect[t_0]$.  $t_0$ is a
    stopping time for $(Z_t)_{t\geq 0}$ as it is the first time at
    which $Z_t = 3m^2-n^{-4}t$.  Therefore, the optional stopping
    theorem says that $\expect[Z_{t_0}]\geq \expect[Z_0]$, as long as
    $\expect[t_0] < \infty$, which we show below.  It
    follows, then, that
    \begin{equation*}
        3m^2 - n^{-4}\expect[t_0] \geq \expect[\phi(X_0)^2] - m^2,
    \end{equation*}
    which gives
    \begin{equation*}
        \expect[t_0] \leq n^4(4m^2 - \expect[\phi(X_0)^2])
            = n^4(\phi(G)^2 - \expect[\phi(X_0)^2])\,,
    \end{equation*}
    as required.

    It remains to establish that $t_0$ has finite expectation.
    Consider a block of $n$ successive stages $X_k, \dots, X_{k+n-1}$.
    If the Moran process has not already reached absorption by $X_k$,
    then $|X_k|\geq 1$.  Consider any sequence of reproductions by
    which a single mutant in $X_k$ could spread through the whole
    graph.  Each transition in that sequence has probability at least
    $n^{-2}$ so the sequence has probability at least $p = (n^{-2})^n$
    and, therefore, the probability of absorption within the block is
    at least this value.  But then the expected number of blocks
    before absorption is at most
    \begin{equation*}
        \sum_{i > 0} (1-p)^{i-1} = \frac{1}{1-(1-p)} = \frac{1}{p}\,.
    \end{equation*}
    and, therefore, $\expect[t_0] < \infty$ as required.
\end{proof}

\begin{corollary}
\label{cor:abs-martingale}
    (i) When $r=1$ and the initial single mutant is chosen uniformly
    at random, the expected absorption time of the Moran process is at
    most $t = \phi(G)^2 n^4\!$.  (ii) For any $\epsilon\in (0,1)$, the
    process reaches absorption within $t/\epsilon$ steps with
    probability at least $1-\epsilon$.
\end{corollary}
\begin{proof}
    The first part is immediate from the previous theorem and the fact
    that $\expect[\phi(X_0)^2]>0$.  The second part follows by
    Markov's inequality.
\end{proof}

When the initial state is a single mutant chosen uniformly at random,
we have
\begin{equation*}
    \expect[\phi(X_0)^2]
        = \frac{1}{n} \sum_{x\in V(G)} \frac{1}{(\deg x)^2}
        < \frac{1}{n} \left(\sum_{x\in V(G)} \frac{1}{\deg x}\right)^{\!2}
        = \ \frac{\phi(G)^2}{n}\,,
\end{equation*}
so little is lost by discarding the $\expect[\phi(X_0)^2]$ term.

 \section{Approximation algorithms}
\label{sec:FPRAS}

We now have all the components needed to present our fully
polynomial randomized approximation schemes (FPRAS) for the problem of
computing the fixation probability of a graph, where $r\geq 1$, and
for computing the extinction probability for all $r>0$.  Recall that
an FPRAS for a function $f$ is a randomized algorithm $g$ that, given
input $X$, gives an output satisfying
\begin{equation*}
    (1-\epsilon)f(X) \leq g(X) \leq (1+\epsilon)f(X)
\end{equation*}
with probability at least $\frac{3}{4}$ and has running time
polynomial in both $|X|$ and~$\frac{1}{\epsilon}$.  Although the value
of $\frac{3}{4}$ is rather low for practical use, the same class of
problems has an FPRAS if we choose any probability $\frac{1}{2}<p<1$
\cite{JVV1986:Randgen}.  Furthermore, the probability that the result
is within a factor of $1\pm\epsilon$ of the true value can be
increased from $\tfrac{3}{4}$ to $1-\delta$ for any positive~$\delta$,
just by taking the median answer from $\bigO(\log \tfrac{1}{\delta})$
runs of the algorithm \cite[Lemma~6.1]{JVV1986:Randgen}.

In the following two theorems, we give algorithms whose running times
are polynomial in $n$, $r$ and $\frac{1}{\epsilon}$.  For the
algorithms to run in time polynomial in the length of the input, and
thus meet the definition of FPRAS, $r$ must be encoded in unary.

\begin{theorem}
    There is an FPRAS for \textsc{Moran fixation}, for $r\geq 1$.
\end{theorem}
\begin{proof}
    The algorithm is as follows.  If $r=1$ then, by
    Lemma~\ref{lemma:fG-one}, we return $\frac{1}{n}$.
    Otherwise, we simulate the Moran process on $G$ for $T =
    \lceil\frac{8r}{r-1}Nn^4\rceil$ steps, $N = \lceil\frac{1}{2}
    \epsilon^{-2} n^2\ln 16\rceil$ times and compute the proportion of
    simulations that reached fixation.  If any simulation has not
    reached absorption (fixation or extinction) after $T$ steps, we
    abort and immediately return an error value.

    Note that each transition of the Moran process can be simulated in
    $\bigO(1)$ time.  Maintaining arrays of the mutant and non-mutant
    vertices allows the reproducing vertex to be chosen in constant
    time and storing a list of each vertex's neighbours allows the
    same for the vertex where the offspring is sent.  Therefore, the
    total running time is $\bigO(NT)$ steps, which is polynomial in
    $n$ and $\frac{1}{\epsilon}$, as required.

    It remains to show that the algorithm operates within the required
    error bounds.  For $i\in \{1, \dots, N\}$, let $X_i=1$ if the
    $i$th simulation of the Moran process reaches fixation and $X_i=0$
    otherwise.  Assuming all simulation runs reach absorption, the
    output of the algorithm is $p = \frac{1}{N}\sum_i X_i$.  By
    Hoeffding's inequality and writing $f = \fGr$, we have
    \begin{equation*}
        \prob[|p - f| > \epsilon f]
            \leq 2\exp(-2 \epsilon^2 f^2 N)
            \leq 2\exp(-f^2 n^2 \ln 16)
            \leq \tfrac{1}{8}\,,
    \end{equation*}
    where the final inequality is because, by
    Corollary~\ref{cor:fG-lowerbound}, $f\geq \frac{1}{n}$.

    Now, the probability that any individual simulation has not
    reached absorption after $T$ steps is at most $\frac{1}{8N}$ by
    Corollary~\ref{cor:abs-time-whp}.  Taking a union
    bound, the probability of aborting and returning an error because
    at least one of the $N$ simulations was cut off before reaching
    absorption is at most $\frac{1}{8}$.  Therefore, with probability
    at least $\frac{3}{4}$, the algorithm returns a value within a
    factor of $1\pm\epsilon$ of $\fGr$.
\end{proof}

Note that this technique fails for disadvantageous mutants ($r<1$)
because there is no analogue of Corollary~\ref{cor:fG-lowerbound} giving
a polynomial lower bound on $\fGr$.  As such, an exponential number of
simulations may be required to achieve the desired error probability.
However, we can give an FPRAS for the extinction probability for all
$r>0$.  Although the extinction probability is just $1-\fGr$, there is
no contradiction because a small relative error in $1-\fGr$ does not
translate into a small relative error in $\fGr$ when $\fGr$ is, itself,
small.

\begin{theorem}
    There is an FPRAS for \textsc{Moran extinction} for all $r>0$.
\end{theorem}
\begin{proof}
    The algorithm and its correctness proof are essentially as in the
    previous theorem.  If $r=1$, we return $1-\tfrac{1}{n}$.  Otherwise,
    we run $N=\lceil\frac{1}{2}\epsilon^{-2} (r+n)^2 \ln 16\rceil$
    simulations of the Moran process on $G$ for $T(r)$ steps each, where 
    \begin{align*}
        T(r) &= \begin{cases}
                \ \lceil\frac{8r}{r-1}Nn^4\rceil     &\text{ if $r>1$} \\
                \ \lceil\frac{8}{1-r}Nn^3\rceil      &\text{ if $r<1$.}
             \end{cases}
    \end{align*}
    If any simulation has not reached absorption within $T(r)$ steps,
    we return an error value; otherwise, we return the proportion $p$
    of simulations that reached extinction.

    Writing $\fbar=1-\fGr$ for the extinction probability, Hoeffding's
    inequality gives
    \begin{equation*}
        \prob[|p-\fbar| > \epsilon \fbar]
            \leq 2\exp(-2 \epsilon^2 \fbar^2 N)
            \leq 2\exp(-\fbar^2 (r+n)^2 \ln 16)
            \leq \tfrac{1}{8}\,,
    \end{equation*}
    with the final inequality because $\fbar\geq \tfrac{1}{r+n}$ by
    Lemma~\ref{lemma:fG-upperbound}.

    The probability that any given simulation run has not reached
    absorption within $T(r)$ steps is at most $\tfrac{1}{8N}$ by
    Corollary~\ref{cor:abs-time-whp} ($r>1$) or
    Corollary~\ref{cor:abs-time-disadv-whp} ($r<1$) so the algorithm meets
    the error bounds with probability at least $\tfrac{3}{4}$ by the
    same argument as before.
\end{proof}

It remains open whether other techniques could lead to an FPRAS for
\textsc{Moran fixation} when $r<1$.
 
\end{document}